\documentclass[11points]{llncs}
 \pdfoutput=1
\usepackage{amsmath,amssymb}
\usepackage{mathtools}
\usepackage{mathtools}
\usepackage{algorithm}
\usepackage{algorithmicx}
\usepackage[noend]{algpseudocode}
\usepackage{algpascal}
\usepackage{color}
\usepackage{cite}
\usepackage{graphicx}
\usepackage{enumerate}
\usepackage{colortbl}
\usepackage{vmargin}

\usepackage{tikz}

\spnewtheorem{observation}{Observation}{\bfseries}{\itshape}
\spnewtheorem{claimN}{Claim}{\bfseries}{\itshape}
\spnewtheorem{remarkLI}{Remark}{\bfseries}{\itshape}

\usepackage{amsmath}
\usepackage{amssymb}
\usepackage{color}
\definecolor{red}{RGB}{255,0,0}
\definecolor{blue}{RGB}{0,0,255}
\definecolor{green}{RGB}{0,255,0}

\newcommand {\abs}[1]  {\left\vert#1\right\vert}
\newcommand {\set}[1]  {\left\{#1\right\}}

\newcommand {\defined} {\stackrel{def} {=}}

\newcommand {\nph}     {\textsc{NP}\textrm{-hard}}

\newcommand {\runningtitle}[1] {\vspace{0.5ex}\noindent{\textbf{\boldmath #1:}}}

\newcommand {\commentfig}[1] {#1}

\DeclareMathOperator{\cs}{cs}

\newcommand{\bb}{{\cal B}}
\newcommand{\maxcut} {\textsc{MaxCut}}
\newcommand{\cutsize}[1] {\cs(#1)}
\newcommand{\vect}[1]{\mathbf{#1}}

\title{The Maximum Cardinality Cut Problem is Polynomial in Proper Interval Graphs}

\author{Arman Boyac{\i}~\inst{1} \and Tinaz Ekim~\inst{1} \and Mordechai Shalom~\inst{2}}

\institute{
Department of Industrial Engineering, Bo\u{g}azi\c{c}i University, Istanbul, Turkey \\
\email{[arman.boyaci, tinaz.ekim]@boun.edu.tr}
\and
TelHai College, Upper Galilee, 12210, Israel\\
\email{cmshalom@telhai.ac.il}
}

\begin{document}
\maketitle
\begin{abstract}
It is known that the maximum cardinality cut problem is $\nph$ even in chordal graphs. In this paper, we consider the time complexity of the problem in proper interval graphs, a subclass of chordal graphs, and propose a dynamic programming algorithm which runs in polynomial-time.
\end{abstract}

\section{Introduction}\label{sec:intro}
A \emph{cut} of a graph $G=(V(G),E(G))$ is a partition of $V(G)$ into two subsets $S, \bar{S}$ where $\bar S=V(G)\setminus S$. The \emph{cut-set} of $(S,\bar{S})$ is the set of edges of $G$ with exactly one endpoint in $S$. The maximum cardinality cut problem ($\maxcut$) is to find a cut with a maximum size cut-set, of a given graph.

$\maxcut$ remains $\nph$ when restricted to the following graph classes: chordal graphs, undirected path graphs, split graphs, tripartite graphs, co-bipartite graphs \cite{Bodlaender00onthe}, unit disk graphs \cite{DK2007} and total graphs \cite{Guruswami1999217}. On the other hand, it was shown that $\maxcut$ can be solved in polynomial-time in planar graphs \cite{hadlock1975finding}, in line graphs \cite{Guruswami1999217}, and the class of graphs factorable to bounded treewidth graphs \cite{Bodlaender00onthe}.

Proper interval graphs are not necessarily planar. They are not necessarily line graphs either, since the graph $\bar{A}$ consisting of 6 vertices is a proper interval graph, but a forbidden subgraph of line graphs \cite{Beineke1970129}. In \cite{BES15-MaxCut-JOCOSpecialIssue}, we have shown that co-bipartite chain graphs are not factorable to bounded treewidth graphs. Since co-bipartite chain graph are proper interval graphs, this result holds for them too. Therefore, the above mentioned results do not imply a polynomial-time algorithm for proper interval graphs.

Despite the existence of polynomial-time algorithms for some subclasses of proper interval graphs (split indifference graphs \cite{Bodlaender2004} and co-bipartite chain graphs \cite{BES15-MaxCut-JOCOSpecialIssue}), the complexity of $\maxcut$ in proper interval graphs was open. In this work, we generalize the dynamic programming algorithm in \cite{BES15-MaxCut-JOCOSpecialIssue} to proper interval graphs using the bubble model of proper interval graphs introduced in \cite{HMP09}.

\section{Preliminaries}\label{sec:prelim}
\runningtitle{Graph notations and terms}
Given a simple graph (no loops or parallel edges) $G=(V(G),E(G))$ and a vertex $v$ of $G$, $uv$ denotes an edge between two vertices $u,v$ of $G$. We also denote by $uv$ the fact that $uv \in E(G)$. We denote by $N(v)$ the set of neighbors of $v$. Two adjacent (resp. non-adjacent) vertices $u,v$ of $G$ are \emph{twins} if $N_G(u) \setminus \set{v} = N_G(v) \setminus \set{u}$. A vertex having degree zero is termed \emph{isolated}, and a vertex adjacent to all other vertices is termed \emph{universal}. For a graph $G$ and $U \subseteq V(G)$, we denote by $G[U]$ the subgraph of $G$ induced by $U$, and $G \setminus U \defined G[V(G) \setminus U]$. For a singleton $\set{x}$ and a set $Y$, $Y + x \defined Y \cup \set{x}$ and $Y - x \defined Y \setminus \set{x}$. A vertex set $U \subseteq V(G)$ is a \emph{clique} (resp. \emph{stable set}) (of $G$) if every pair of vertices in $U$ is adjacent (resp. non-adjacent). We denote by $n$ be the number of vertices of $G$.

\runningtitle{Some graph classes}
A graph is \emph{bipartite} if its vertex set can be partitioned into two independent sets $V,V'$. We denote such a graph as $B(V,V',E)$ where $E$ is the edge set. A graph $G$ is \emph{co-bipartite} if it is the complement of a bipartite graph, i.e. $V(G)$ can be partitioned into two cliques $K, K'$. We denote such a graph as $C(K,K',E)$ where $E$ is the set of edges that have exactly one endpoint in $K$.

A \emph{bipartite chain graph} is a bipartite graph $G=B(V,V',E)$ where $V$ has a nested neighborhood ordering, i.e. its vertices can be ordered as $v_1,v_2,\ldots$  such that $N_G(v_1) \subseteq N_G(v_2) \subseteq \cdots$. $V$ has a nested neighborhood ordering if and only if $V'$ has one \cite{yannakakis1981node}. Theorem 2.3 of \cite{HPS90} implies that if $G=B(V,V',E)$ is a bipartite chain graph with no isolated vertices, then the number of distinct degrees in $V$ is equal to the number of distinct degrees in $V'$.

A co-bipartite graph $G=C(K,K',E)$ is a \emph{co-bipartite chain} (also known as co-chain) graph if $K$ has a nested neighborhood ordering \cite{heggernes2007linear}. Since $K \subseteq N_G(v)$ for every $v \in K$, the result for chain graphs implies that $K$ has a nested neighborhood ordering if and only if $K'$ has such an ordering.

A graph $G$ is \emph{interval} if its vertices can be represented by intervals on a straight line such that two vertices are adjacent in $G$ if the corresponding intervals are intersecting. An interval graph is \emph{proper} (resp. \emph{unit}) if there is a interval representation such that no interval properly contains another (resp. every interval has unit length). It is known that the proper interval graph class is equivalent to unit interval graph class \cite{Bogart199921}.

\runningtitle{Cuts}
We denote a cut of a graph $G$ by one of the subsets of the partition. $E(S,\bar{S})$ denotes the \emph{cut-set} of $S$, i.e. the set of the edges of $G$ with exactly one endpoint in $S$, and $\cutsize{S} \defined \abs{E(S,\bar{S})}$ is termed the \emph{cut size} of $S$. A maximum cut of $G$ is one having the biggest cut size among all cuts of $G$. We refer to this size as the \emph{maximum cut size} of $G$. Clearly, $S$ and $\bar{S}$ are dual; we thus can replace $S$ by $\bar{S}$ and $\bar{S}$ by $S$ everywhere. In particular, $E(S,\bar{S})=E(\bar{S},S)$, and $\cutsize{S}=\cutsize{\bar{S}}$.

\runningtitle{Bubble models}
A \emph{2-dimensional bubbles structure} $\bb$ for a finite non-empty set $A$ is a 2-dimensional arrangement of bubbles $\set{B_{i,j}~|~j \in [k], i \in [r_j]}$ for some positive integers $k, r_1, \ldots r_k$, such that $\bb$ is a near-partition of $A$. That is, $A = \cup \bb$ and the sets $B_{i,j}$ are pairwise disjoint, allowing for the possibility of $B_{i,j}=\emptyset$ for arbitrarily many pairs $i,j$. For an element $a \in A$ we denote by $i(a)$ and $j(a)$ the unique indices such that $a \in B_{i(a),j(a)}$.

Given a bubble structure $\bb$, the graph $G(\bb)$ defined by $\bb$ is the following graph:
\begin{enumerate}[i)]
\item $V(G(\bb))=\cup \bb$, and
\item $uv \in E(G(\bb))$ if and only if one of the following holds:
\begin{itemize}
\item $j(u)=j(v)$,
\item $j(u)=j(v)+1$ and $i(u) < i(v)$.
\end{itemize}
\end{enumerate}
$\bb$ is a \emph{bubble model} for $G(\bb)$.

A \emph{compact representation} for a bubble model is an array of \emph{columns} each of which contains a list of non-empty bubbles, and each bubble contains its row number in addition to the vertices in this bubble.

\begin{theorem}\cite{HMP09}\label{thm:Bubbles}
\begin{enumerate}[i)]
\item A graph is a proper interval graph if and only if it has a bubble model.
\item A bubble model for a graph on $n$ vertices contains $O(n^2)$ bubbles and it can be computed in $O(n^2)$ time.
\item A compact representation of a bubble model for a graph on $n$ vertices can be computed in $O(n)$ time.
\end{enumerate}
\end{theorem}

Note that the set of vertices in two consecutive columns in $\bb$ induces a co-bipartite chain graph. In other words, a proper interval graph can be seen as a chain of co-bipartite chain graphs, see Figure \ref{fig:BubbleModel}. Using this observation, we generalize our result in \cite{BES15-MaxCut-JOCOSpecialIssue}. To keep the analysis simpler, we use the standard representation of the bubble model, since using the compact representation does not improve the overall running time of the algorithm.

\begin{figure}
\centering
\commentfig{\includegraphics[width=0.7\textwidth]{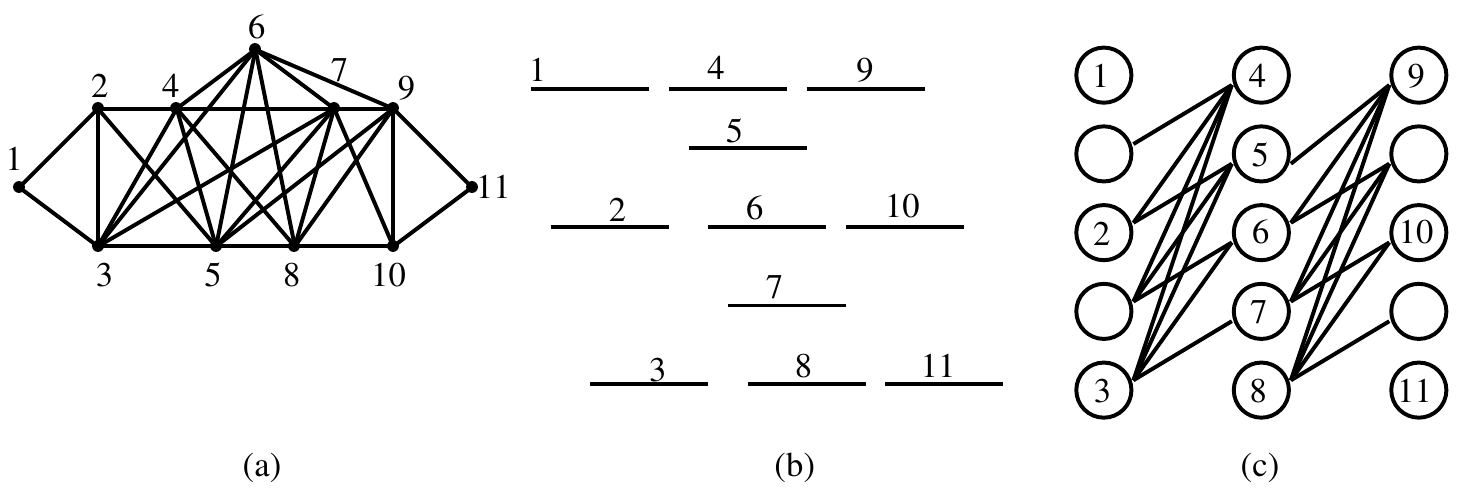}}
\caption{(a) A proper interval graph $G$, (b) an interval representation of $G$, (c) a bubble model of $G$}
\label{fig:BubbleModel}
\end{figure}

\section{Algorithm for Proper Interval Graphs}
\newcommand{\zero}{\vect{0}}
\newcommand{\one}{\vect{1}}
\newcommand{\sumvect}[1]{\sum \vect{#1}}
\newcommand{\pigdynalg}{\textsc{ProperIntervalDynamicProgramming}}

Let $G$ be a proper interval graph and $\bb=\set{B_{i,j}~|~j \in [k], i \in [r_j]}$ a bubble model for it, where $k \geq 1$ is the number of columns and $r_1, \ldots, r_k$ are the number of rows in the columns. We note that two vertices in the same bubble are twins. For a cut $S$ of $G$ and a bubble $B_{i,j}$ we denote $b_{i,j}=\abs{B_{i,j}}$, $s_{i,j} = \abs{B_{i,j} \cap S}$, and $\bar{s}_{i,j} = \abs{B_{i,j} \setminus S}$. Clearly, $0 \leq s_{i,j}, \bar{s}_{i,j} \leq b_{i,j}$, and the cut $S$ is uniquely defined by the matrix $\vect{s}$. We use $C_j$ as a shorthand for $\cup_{i=1}^{r_j} B_{i,j}$. We also denote $c_j=\abs{C_j}$, $s_j=\abs{C_j \cap S}$, and $\bar{s}_j=\abs{C_j \setminus S}$.

In what follows, we first show a recurrence relation for the size of a maximum cardinality cut in proper interval graphs and then develop a dynamic programming algorithm based on it.

\begin{theorem}\label{thm:recurrenceRelation}
The maximum cut size of a proper interval graph $G$ with a bubble model $\bb$ is
$F_{0,0}(0, 0)$
where $F_{i,j}(x, x')$ is given by the following recurrence relation:
\begin{eqnarray}
& F_{0,k+1}(0,0) = 0,\label{eqn:RecurrenceBase1} \\
& \forall j \in [0,k], & F_{0,j}(0,0) = \max_{\substack{s_{j+1} \in [0,b_{j+1}], s_{j+2} \in [0,b_{j+2}]}} F_{r_{j+1},j+1}(s_{j+1},s_{j+2}), \label{eqn:RecurrenceBase2}\\
& \forall j \in [k], \forall i \in [r_k], & F_{i,j}(x, x') = b_{i,j} (x+x') + \\
& & \max_{\substack{ L_{i,j} \leq s_{i,j} \leq U_{i,j} \\ L_{i,j+1} \leq s_{i,j+1} \leq U_{i,j+1}}}
\begin{cases} F_{i-1,j}(x - s_{i,j}, x' - s_{i,j+1})  - b_{i,j} ~s_{i,j+1}
 + \nonumber \\
s_{i,j} \left( \sum_{i'=1}^{i-1} (b_{i',j}+b_{i',j+1}) - 2 x - 2x' + s_{i,j} + 2 s_{i,j+1} \right)\\
\end{cases}\label{eqn:RecurrenceStep}
\end{eqnarray}
where $L_{i,j} = \max \left(0, x - \displaystyle \sum_{i'=1}^{i-1} b_{i',j} \right)$, $U_{i,j} = \min \left( b_{i,j},x \right)$, $L_{i,j+1} = \max \left(0, x' - \displaystyle \sum_{i'=1}^{i-1} b_{i',j+1} \right)$ and $U_{i,j+1} = \min \left( b_{i,j+1},x \right)$.
\end{theorem}

\begin{proof}
For convenience we add to the bubble representation, three columns, namely columns $0,k+1$ and $k+2$ with no bubbles, i.e. $r_0=r_{k+1}=r_{k+2}=0$. Throughout this proof, $j \in [0,k+2]$ and $i \in [0,r_j]$. Whenever we refer to a non-existent bubble $B_{i,j}$ (by allowing $i > r_j$) we assume $b_{i,j}=0$.

We denote by $G_{i,j}$ the subgraph of $G$ induced by the vertices in the first $i$ rows of column $j$ and all the vertices in the columns from $j+1$ to $k+1$. Formally, $G_{i,j}=G \left[ \left(\cup_{i'=1}^i B_{i',j} \right) \bigcup \left( \cup_{j'=j+1}^{k+2} B_{j'} \right) \right]$. We observe that a) $G=G_{0,0}$, b) $G_{0,j}=G_{r_{j+1},j+1}$,  and c) $G_{0,k}=G_{0,k+1}$ are empty graphs.

$F_{i,j}(x,x')$ denotes the maximum cut size among all cuts $S$ of $G_{i,j}$ such that $\sum_{i'=1}^i s_{i,j} = x$ and $\sum_{i'=1}^i s_{i,j+1} = x'$,  i.e.,
\[
F_{i,j}(x,x')=\max \set{\cutsize{S}| S \subseteq V(G_{i,j}), \sum_{i'=1}^i s_{i',j} = x, \sum_{i'=1}^i s_{i',j+1} = x'}.
\]
For $i=0$, since $\sum_{i'=1}^0 s_{i',j}=\sum_{i'=1}^0 s_{i',j}=0$ the only relevant arguments to $F_{0,j}$ are $x=x'=0$. Recalling that $G=G_{0,0}$, it is clear that the maximum cut size of $G$ is $F_{0,0}(0,0)$. We now provide a recurrence relation for $F_{i,j}(x,x')$.

\begin{figure}
\centering
\commentfig{\includegraphics[width=0.7\textwidth]{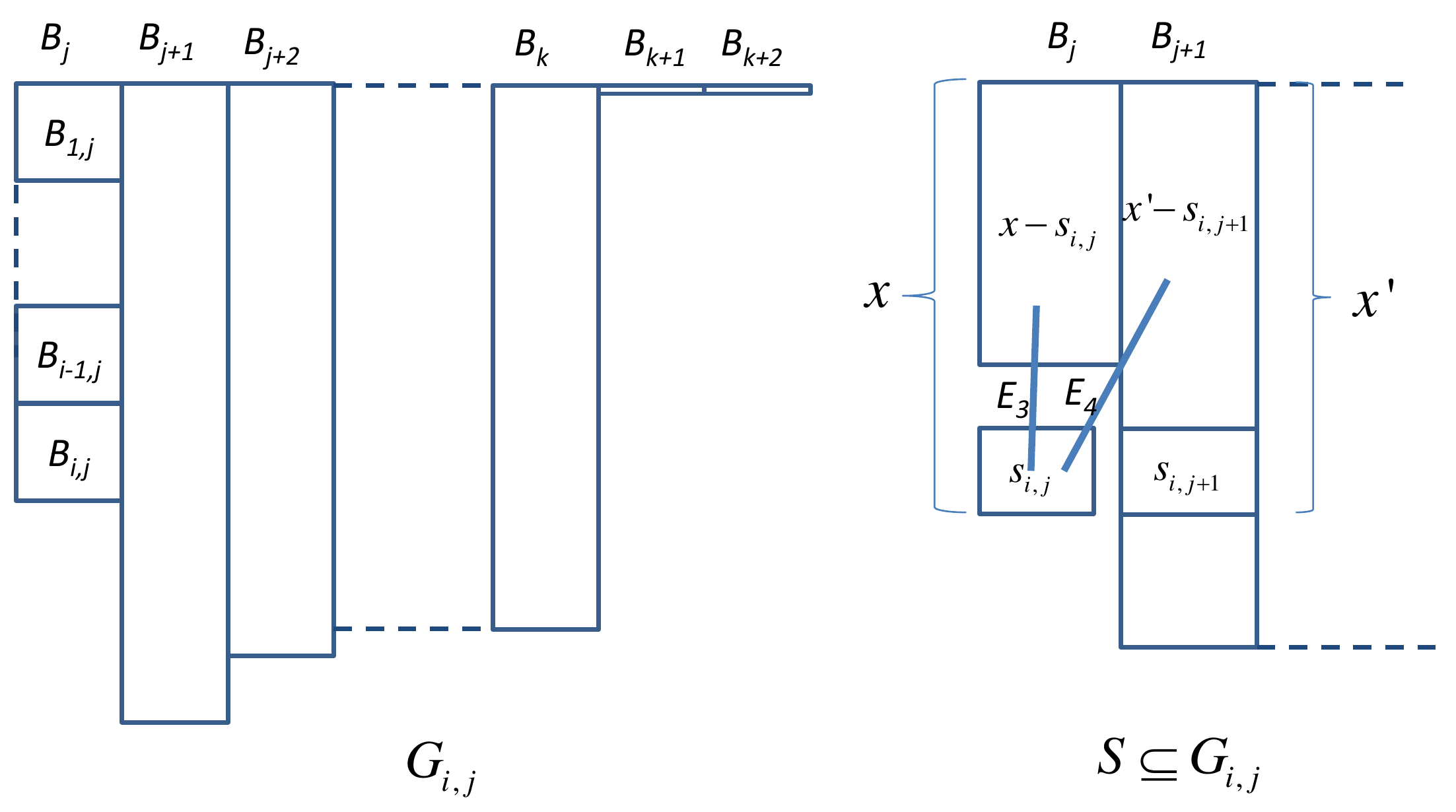}}
\caption{The graph $G_{i,j}$ and the decomposition of a cut $S$ of $G_{i,j}$ having $x$ vertices in $\cup_{i'=1}^i B_{i,j}$ and $x'$ vertices in $\cup_{i'=1}^i B_{i,j+1}$.}
\label{fig:DynamicProgramming}
\end{figure}

Since $G_{i,k+1}$ is the empty graph, (\ref{eqn:RecurrenceBase1}) clearly holds. We now compute the value $F_{0,j}(0,0)$. Recall that $G_{0,j}=G_{r_{j+1},j+1}$. Therefore, the best cut of $G_{0,j}$ can be computed by trying all possible arguments of $F_{r_{j+1},j+1}$ and choosing the maximum. Therefore, (\ref{eqn:RecurrenceBase2}) holds.

For the following discussion, refer to Figure \ref{fig:DynamicProgramming}. Consider a cut $S$ of the subgraph $G_{i,j}$ such that $\sum_{i'=1}^i s_{i',j}=x$ and $\sum_{i'=1}^i s_{i',j+1}=x'$. The graph $G_{i,j}$ can be partitioned into the subgraph $G_{i-1,j}$ and the complete graph $G[B_{i,j}]$. We partition the edges of $G_{i,j}$ into four sets: a) the edges of $G_{i-1,j}$, b) the edges of $G[B_{i,j}]$, c) the edges with one endpoint in $B_{i,j}$ and one in $\cup_{i'=1}^{i-1}B_{i',j}$, and d) the edges with one endpoint in $B_{i,j}$ and one in $\cup_{i'=1}^{i-1}B_{i',j+1}$. Accordingly, the edges of $E(S,\bar{S})$ can be partitioned into four sets:
\begin{align*}
E_S^1 & = E(S,\bar{S}) \cap E(G_{i-1,j}), \\
E_S^2 & = \left \{uv \in E(S,\bar{S}) | u,v \in B_{i,j} \right\}, \\
E_S^3 & = \left \{uv \in E(S,\bar{S}) | u \in B_{i,j}, v \in \cup_{i'=1}^{i-1} B_{i',j} \right \}, \\
E_S^4 & = \left \{uv \in E(S,\bar{S}) | u \in B_{i,j}, v \in \cup_{i'=1}^{i-1} B_{i',j+1} \right \}.
\end{align*}
Let $\underbar{S}$ be the cut that $S$ induces on $G_{i-1,j}$, and note that by the definition of a bubble representation, every vertex of $B_{i,j}$ is adjacent to every vertex of $\cup_{i'=1}^{i-1} B_{i',j}$ and to every vertex of $\cup_{i'=1}^{i-1} B_{i',j+1}$. Then, the sizes of the above sets are
\begin{align*}
\abs{E_S^1} & = \cutsize{\underbar{S}} \\
\abs{E_S^2} & = s_{i,j} \cdot \bar{s}_{i,j} \\
\abs{E_S^3} & = s_{i,j} \left( \sum_{i'=1}^{i-1} b_{i',j} - (x - s_{i,j}) \right) + \bar{s}_{i,j} (x - s_{i,j}) \\
\abs{E_S^4} & = s_{i,j} \left( \sum_{i'=1}^{i-1} b_{i',j+1} - (x' - s_{i,j+1}) \right) + \bar{s}_{i,j} (x' - s_{i,j+1}).
\end{align*}
We note that $\abs{E_S^1}+\abs{E_S^2}+\abs{E_S^3}$ does not depend on $S$, but only on the graph $G_{i,j}$ and the four values $x$,$x'$,$s_{i,j}$ and $s_{i,j+1}$. Then, maximizing $\cutsize{S}$ for any fixed choice of these values is equivalent to maximizing $\cutsize{\underbar{S}}$ for these values. We now observe that $\sum_{i'=1}^{i-1} \underbar{s}_{i',j}=x-s_{i,j}$ and $\sum_{i'=1}^{i-1} \underbar{s}_{i',j+1}=x'-s_{i',j+1}$. Then the maximum of $\cutsize{\underbar{S}}$ is $F_{i-1,j}(x - s_{i,j}, x' - s_{i,j+1})$.
Then $F_{i,j}(x,x')$ is the maximum of
\[
F_{i-1,j}(x - s_{i,j}, x' - s_{i,j+1}) + \abs{E_S^2} + \abs{E_S^3} + \abs{E_S^4}
\]
over all possible choices of $s_{i,j}, s_{i,j+1}$.
As for the possible choices of $s_{i,j}, s_{i,j+1}$, we recall that $0 \leq s_{i,j} \leq b_{i,j}$ and $0 \leq s_{i,j+1} \leq b_{i,j+1}$. Similarly, $0 \leq x - s_{i,j} \leq \sum_{i'=1}^{i-1} b_{i',j}$ and $0 \leq x' - s_{i,j+1} \leq \sum_{i'=1}^{i-1} b_{i',j+1}$. Substituting the values for $\abs{E_S^2}$, $\abs{E_S^3}$, and $\abs{E_S^4}$ and rearranging terms, we get equation (\ref{eqn:RecurrenceStep}).
\qed
\end{proof}

\alglanguage{pseudocode}
\begin{algorithm}[htbp]
\caption{$\pigdynalg$}\label{alg:DynAlg}
\begin{algorithmic}[1]
\Require {$G$ is a proper interval graph}
\Ensure {The maximum cardinality cut size of $G$}
\Statex

\State $\bb \gets$ a bubble representation of $G$ with $O(\abs{V(G)}^2)$ bubbles.

\For {$i = 1$ \textbf{to} $r_k$}
    \State $F_{i,k+1}(0,0) \gets 0$
\EndFor
\Statex

\For {$j = k$ \textbf{to} $1$}
\State $F_{0,j}(0,0) \gets \Call{SummarizeColumn}{j+1}$
    \For {$i = 1$ \textbf{to} $r_j$}
    \For {$x=0$ \textbf{to} $\sum_{i'=0}^i b_{i',j}$}
        \For {$x'=0$ \textbf{to} $\sum_{i'=0}^i b_{i',j+1}$}
            \State $F_{i,j}(x,x') \gets $\Call{CalculateOpt}{$i,j,x,x'$}.
        \EndFor
    \EndFor
    \EndFor
\EndFor
\State \Return $F_{0,0}(0,0)$.
\Statex

\Function{SummarizeColumn}{$j$}
\State $max \gets 0$
\For{$s_j \gets 0$ \textbf{to} $c_j$}
    \For{$s_{j+1} \gets 0$ \textbf{to} $c_{j+1}$}
        \State $val \gets F_{r_j,j}(s_j,s_{j+1})$
        \If {$val > max$}
            \State $max \gets val$.
        \EndIf
    \EndFor
\EndFor
\State \Return $max$.
\EndFunction
\Statex

\Function{CalculateOpt}{$i, j, x, x'$}
\State $b \gets \sum_{i'=0}^{i-1} b_{i',j}$
\State $b' \gets \sum_{i'=0}^{i-1} b_{i',j+1}$
\State $max \gets 0$
\For{$s_{i,j} \gets \max(0,x-b)$ \textbf{to} $\min (b_{i,j},x)$}
    \For{$s_{i,j+1} \gets \max(0,x'-b')$ \textbf{to} $\min (b',x')$}
        \State $val \gets F_{i-1,j}(x-s_{i,j},x'-s_{i,j+1}) - b_{i,j} \cdot s_{i,j+1} + s_i \cdot (b+b'-2x-2x'+s_{i,j}+2s_{i,j+1})$.
        \If {$val > max$}
            \State $max \gets val$.
        \EndIf
    \EndFor
\EndFor
\State \Return $max + b_{i,j} \cdot (x + x')$.
\EndFunction
\end{algorithmic}
\end{algorithm}

\begin{theorem}
$\pigdynalg$ is an $O(n^4)$ algorithm for $\maxcut$ in proper interval graphs.
\end{theorem}

\begin{proof}
$\pigdynalg$ calculates the recurrence relation described in Theorem \ref{thm:recurrenceRelation} through dynamic programming, by scanning the columns in descending order, and the rows of each column in ascending order. Therefore, its correctness follows from Theorem \ref{thm:recurrenceRelation}.

By Theorem \ref{thm:Bubbles}, a bubble representation for $G$ can be computed in $O(n^2)$ time. The running time of function \textsc{SummarizeColumn}($j$) is $O(c_j \cdot c_{j+1})$. Summing up for all columns we get $O(\sum_{j=1}^k c_j \cdot c_{j+1}) = O\left( \left( \sum_{j=1}^k c_j \right)^2 \right)=O(n)^2$.

The running time of function \textsc{CalculateOpt}($i,j,x,x'$) is $O(b_{i,j} \cdot b_{i,j+1})$. It remains to compute the total running time of all \textsc{CalculateOpt} invocations. This time is proportional to at most
\begin{align*}
&&\sum_{i,j~\textrm{s.t.}~b_{i,j}>0} c_j \cdot c_{j+1} \cdot b_{i,j} \cdot b_{i,j+1} = \sum_{j=1}^k \left( c_j \cdot c_{j+1} \sum_{i~\textrm{s.t.}~b_{i,j}>0} b_{i,j} \cdot b_{i,j+1} \right)\\
& \leq & \sum_{j=1}^k c_j \cdot c_{j+1} \cdot c_j \cdot c_{j+1} = \sum_{j=1}^k c_j^2 \cdot c_{j+1}^2 \leq \left( \sum_{j=1}^k c_j^2 \right) \left( \sum_{j=1}^k c_{j+1}^2 \right) \leq n^4.
\end{align*}
\qed
\end{proof}


\bibliographystyle{abbrvwithurl}	
\bibliography{Approximation,GraphTheory,Mordo}

\end{document}